\documentclass[preprint]{ptephy_v1}

\usepackage{amsmath,graphics}

\newtheorem{proposition}{Proposition}

\newcommand{\be}{\begin{equation}}
\newcommand{\ee}{\end{equation}}
\newcommand{\vect}[1]{\boldsymbol{#1}}

\renewcommand{\theenumi}{(\arabic{enumi})}

\begin{document}

\title{Solutions of random-phase approximation equation
for positive-semidefinite stability matrix}

\author{\name{H. Nakada}{\ast}}

\affil{Department of Physics, Graduate School of Science, Chiba University,\\
Yayoi-cho 1-33, Inage, Chiba 263-8522, Japan
\email{nakada@faculty.chiba-u.jp}}

\begin{abstract}%
It is mathematically proven that,
if the stability matrix $\mathsf{S}$ is positive-semidefinite,
solutions of the random-phase approximation (RPA) equation are all physical
or belong to Nambu-Goldstone (NG) modes,
and the NG-mode solutions may form Jordan blocks of $\mathsf{N\,S}$
($\mathsf{N}$ is the norm matrix)
but their dimension is not more than two.
This guarantees that the NG modes in the RPA can be separated out
via canonically conjugate variables.
\end{abstract}

\subjectindex{RPA, stability, Nambu-Goldstone mode}

\maketitle

The random-phase approximation (RPA) is widely used
as describing excitation properties
on top of mean-field (MF) solutions.
In Ref.~\cite{ref:Nak16},
I have mathematically argued properties of solutions of the RPA equation,
based on two types of dualities
which are named UL- and LR-dualities.
The solutions have been classified into the five categories
as disclosed by the dualities,
in Prop.~2 of Ref.~\cite{ref:Nak16}.
It has been also reconfirmed that,
if the stability matrix $\mathsf{S}$ is positive-definite,
the solutions are all physical, belonging to Class~(1) of Prop.~2,
which was already verified in Ref.~\cite{ref:Thou61}.
Its opposite has been proven as well.
However, spontaneous symmetry breakdown (SSB) necessarily occurs
for the MF solution in localized self-bound systems
like atomic nuclei~\cite{ref:Thou60}.
Individual SSB leads to a Nambu-Goldstone (NG) mode,
and therefore $\mathsf{S}$ is quite generally positive-semidefinite
in physical cases, rather than positive-definite.

A method to handle NG modes has been established~\cite{ref:Thou61,ref:TV62,ref:RS80},
which seems valid as long as the NG mode
corresponds to physical degrees-of-freedom (d.o.f.).
This method presumes that each NG mode forms a two-dimensional Jordan block
of $\mathsf{N\,S}$, where $\mathsf{N}$ is the norm matrix.
On the contrary,
the dualities do not limit dimension of Jordan blocks,
as exemplified in Appendix~C.5 of Ref.~\cite{ref:Nak16}.
To my best knowledge, there have been no rigorous arguments
that elucidate dimension of Jordan blocks for the NG modes,
despite an argument for restricted cases~\cite{ref:Ner09}.
In this paper, I shall prove what dimensionality is possible
for Jordan blocks associated with the NG-mode solutions
(\textit{i.e.}, Class~(5) in Prop.~2 of Ref.~\cite{ref:Nak16}),
when $\mathsf{S}$ is positive-semidefinite.
Section, Appendix or Proposition number referred to in the text
signifies that of Ref.~\cite{ref:Nak16}, even if not mentioned.

The RPA equation is expressed as
\be \mathsf{S}\,\vect{x}_\nu=\omega_\nu\mathsf{N}\,\vect{x}_\nu\,;\quad
 \mathsf{N}:=\begin{pmatrix} 1&0\\ 0&-1 \end{pmatrix}\,.
\label{eq:RPAeq-b}\ee
The stability matrix $\mathsf{S}$ is a square matrix satisfying
\be \mathsf{S}=\mathsf{S}^\dagger\,,\quad
 \mathsf{\Sigma}_x\,\mathsf{S}^\ast\,\mathsf{\Sigma}_x = \mathsf{S}\,;\quad
 \mathsf{\Sigma}_x := \begin{pmatrix} 0&1\\1&0 \end{pmatrix}\,,
\label{eq:S-prop}\ee
and its dimensionality is denoted by $2D$.
It is noted that Eq.~(\ref{eq:RPAeq-b}) can read
an eigenvalue problem of $\mathsf{N\,S}$.
\begin{proposition}\label{theor:semi-stab}
If the stability matrix $\mathsf{S}$ is positive-semidefinite,
solutions of the RPA equation are constrained to those of Classes~(1) and (5)
in Prop.~2,
and dimensionality of the Jordan blocks for the NG-mode solutions
does not exceed two.
\end{proposition}
\begin{proof}
Provided that $\mathsf{S}$ is positive-semidefinite.
As shown in Sec.~3.3,
$\tilde{\mathsf{S}}:=\mathsf{S}^{1/2}\,\mathsf{N\,S}^{1/2}$ is hermitian,
and gives an eigenvalue problem almost equivalent to that of $\mathsf{N\,S}$.
To be precise, eigenvalues of $\tilde{\mathsf{S}}$
match those of $\mathsf{N\,S}$ up to their degeneracies,
and eigenvectors of $\mathsf{N\,S}$ are those of $\tilde{\mathsf{S}}$,
although the opposite of the latter part is not fulfilled
for Jordan blocks of $\mathsf{N\,S}$.
Quite similarly to the proof for Prop.~5,
$\tilde{\mathsf{S}}$ is diagonalizable
by a certain matrix $\mathsf{Y}$,
$\tilde{\mathsf{S}}\,\mathsf{Y}=\mathsf{Y}\,\tilde{\mathsf{\Omega}}$,
where $\tilde{\mathsf{\Omega}}$ is a diagonal matrix.
The eigenvalues in $\tilde{\mathsf{\Omega}}$ are all real
and paired as $\pm\omega_\nu$ (see Prop.~1 of Ref.~\cite{ref:Nak16}).
We take ${\displaystyle\tilde{\mathsf{\Omega}}:=\begin{pmatrix}
 \mathrm{diag.}(\omega_\nu)&0\\ 0&-\mathrm{diag.}(\omega_\nu) \end{pmatrix}}$
so that $0\leq\omega_1\leq\omega_2\leq\cdots\leq\omega_D$.
In correspondence to this ordering, we write
$\mathsf{Y}=(\vect{y}_1,\cdots,\vect{y}_{2D})$.
Let us denote dimensionality of $\mathrm{Ker}(\tilde{\mathsf{S}})$
(the kernel of $\tilde{\mathsf{S}}$) by $\tilde{K}$;
\textit{i.e.}, $\omega_1=\omega_2=\cdots=\omega_{\tilde{K}}=0$.
For $\tilde{K}<\nu\leq D$, $(\omega_\nu,\mathsf{N\,S}^{1/2}\vect{y}_\nu)$
provides an eigensolution of $\mathsf{N\,S}$,
which is normalizable because
\be (\mathsf{N\,S}^{1/2}\vect{y}_\nu)^\dagger\,\mathsf{N}\,
 (\mathsf{N\,S}^{1/2}\vect{y}_\nu)
 = \vect{y}_\nu^\dagger\,\tilde{\mathsf{S}}\,\vect{y}_\nu
 = \omega_\nu\,\vect{y}_\nu^\dagger\,\vect{y}_\nu > 0\,. \ee
These solutions belong to Class~(1),
and the others are obviously belong to Class~(5).\\
We now focus on the NG-mode solutions;
\be \mathsf{S}\,\vect{x}_\nu = \vect{0}\,. \label{eq:null} \ee
Equation~(\ref{eq:null}) indicates $\vect{x}_\nu\in\mathrm{Ker}(\mathsf{S})$,
as well as that it is a NG-mode solution of the RPA equation.
Note that $\mathrm{Ker}(\mathsf{S})\subseteq\mathrm{Ker}(\tilde{\mathsf{S}})$.
Suppose that $\vect{\xi}_1^{(\nu)}=\vect{x}_\nu$ forms a Jordan block.
We then have, by setting $\omega_\nu=0$ and $k=1$
in Eq.~(11) of Ref.~\cite{ref:Nak16},
\be \mathsf{S}\,\vect{\xi}_2^{(\nu)}
 = ic_1^{(\nu)}\,\mathsf{N}\,\vect{x}_\nu\,,
\label{eq:Jordan}\ee
where $c_1^{(\nu)}\,(\in\mathbf{C})\ne 0$.
Since $\vect{\xi}_2^{(\nu)}\notin\mathrm{Ker}(\mathsf{S})$,
multiplication of Eq.~(\ref{eq:Jordan}) by $\vect{\xi}_2^{(\nu)\dagger}$ yields
\be \vect{\xi}_2^{(\nu)\dagger}\,\mathsf{S}\,\vect{\xi}_2^{(\nu)}
 = ic_1^{(\nu)}\,\vect{\xi}_2^{(\nu)\dagger}\,\mathsf{N}\,\vect{x}_\nu > 0\,.
\label{eq:S_22} \ee
From Prop.~3 and relevant arguments,
$\vect{\xi}_2^{(\nu)\dagger}\,\mathsf{N}\,\vect{x}_\nu\ne 0$ enables us
to identify $\vect{\xi}_2^{(\nu)}$ to be LR-dual
to $\vect{\xi}_1^{(\nu)}\,(=\vect{x}_\nu)$,
deriving that this Jordan block is self LR-dual (in the respect of Prop.~3)
and $d_\nu=2$ ($d_\nu$ is the dimension of the Jordan block).
\end{proof}

While the above proposition does not exclude the possibility
of a pair of $d_\nu=1$ NG modes,
they result in a trivial case in which both of the canonically conjugate d.o.f.
are not included in the RPA Hamiltonian (\textit{i.e.}, $\mathsf{S}$).
From Prop.~6 the two-dimensional Jordan blocks for the NG-mode solutions
can be made doubly self dual.
Notice that, for self UL-dual basis vectors,
`canonical conjugacy' corresponds to the LR-duality.
Decomposition of the RPA Hamiltonian in terms of conjugate variables
is accomplished by the projectors discussed in Sec.~4\footnote
{The example of $J_\pm$ for the SSB with respect to the rotation
in Sec.~5.3 of Ref.~\cite{ref:Nak16} was not appropriate,
in which the Jordan blocks corresponding to $J_\pm$ were claimed
not to be self LR-dual.
They are self LR-dual in truth, while not self UL-dual.}.
Thus, as long as $\mathsf{S}$ is positive-semidefinite,
it is mathematically guaranteed that the prescription
to separate out the NG modes
via canonically conjugate variables~\cite{ref:Thou61,ref:TV62,ref:RS80,ref:Row70},
which is equivalent to Eqs.~(\ref{eq:Jordan},\ref{eq:S_22}) above,
is applicable.
Then, for $\mathcal{W}_{[\nu]}=\mathrm{Ker}(\tilde{\mathsf{S}})$,
$\mathsf{S}_{[\nu]^{-1}}$ is positive-definite (see Sec.~4 for the notation).

\section*{Acknowledgment}

The author is grateful to K.~Neerg\aa rd
for drawing attention to the subject of this paper.
This work is financially supported in part
by JSPS KAKENHI Grant Number~24105008 and Grant Number~16K05342.



\begin{thebibliography}{99}
\bibitem{ref:Nak16} H. Nakada, Prog. Theor. Exp. Phys. (2016), in press
 [e-Print archive 1604.07972].
\bibitem{ref:Thou61} D.J. Thouless, Nucl. Phys. \textbf{22}, 78 (1961).
\bibitem{ref:Thou60} D.J. Thouless, Nucl. Phys. \textbf{21}, 225 (1960).
\bibitem{ref:TV62} D.J. Thouless and J.G. Valatin,
 Nucl. Phys. \textbf{31}, 221 (1962).
\bibitem{ref:RS80} P. Ring and P. Schuck,
 \textit{The Nuclear Many-Body Problem} (Springer-Verlag, 1980).
\bibitem{ref:Ner09} K. Neerg\aa rd, Phys. Rev. C \textbf{80}, 044313 (2009).
\bibitem{ref:Row70} D.J. Rowe, \textit{Nuclear Collective Motion}
 (Methuen, 1970).

\end{thebibliography}
%

\vfill\pagebreak

\end{document}